\documentclass{amsart}
\usepackage{amsfonts}
\usepackage{amsmath,amscd,lscape}
\usepackage{amsthm}
\usepackage{amssymb}
\usepackage{latexsym}
\usepackage{mathrsfs}
\usepackage{tikz}

\newcommand{\nc}{\newcommand}
\nc{\ben}{\begin{eqnarray}}
\nc{\een}{\end{eqnarray}}

\newcommand{\id}{\mathrm{id}}

\newcommand{\beqa}{\begin{eqnarray}}
\newcommand{\eeqa}{\end{eqnarray}}

\nc{\Z}{{\bold Z}}
\nc\cW{\cal W}
\nc\cG{\cal G}

\newcommand{\taub}{{\overline{\tau}}}
\newcommand{\Phib}{{\overline{\Phi}}}

\usepackage[curve,matrix,arrow,color]{xy}

\usepackage{color}

\newtheorem{lem}{Lemma}[section]

\newtheorem{prop}{Proposition}[section]

\newtheorem{defn}{Definition}[section]
\newtheorem{thm}{Theorem}

\newtheorem{rem}{Remark}

\newcommand{\cO}{\mathscr{O}}

\newcommand{\ZZ}{\mathbb{Z}}

\newcommand{\cal}{\mathcal}

\numberwithin{equation}{section}
\begin{document}

\title[FRT and Askey-Wilson algebras]{FRT presentation of classical Askey-Wilson algebras}
\author{Pascal Baseilhac$^{*}$}
\address{$^*$ Institut Denis-Poisson CNRS/UMR 7013 - Universit\'e de Tours - Universit\'e d'Orl\'eans
Parc de Grammont, 37200 Tours, 
FRANCE}
\email{pascal.baseilhac@idpoisson.fr}

\author{Nicolas Cramp\'e$^{\dagger}$}
\address{$^\dagger$ Laboratoire Charles Coulomb (L2C), Univ Montpellier, CNRS, Montpellier, France.}
\email{nicolas.crampe@umontpellier.fr}

\begin{abstract}
Automorphisms of the infinite dimensional Onsager algebra are introduced. Certain quotients of the Onsager algebra are formulated using a polynomial in these automorphisms. In the simplest case, the quotient coincides with 
the classical analog of the Askey-Wilson algebra. In the general case, generalizations of the classical Askey-Wilson algebra
are obtained. The corresponding class of solutions of the non-standard classical Yang-Baxter
algebra are constructed, from which a generating function of elements in the commutative subalgebra is
derived. We provide also another presentation of the Onsager algebra and of the classical Askey-Wilson
algebras.
\end{abstract}

\maketitle

\vskip -0.5cm

{\small MSC:\ 81R50;\ 81R10;\ 81U15.}

{{\small  {\it \bf Keywords}: Onsager algebra;  Non-standard Yang-Baxter algebra;  Askey-Wilson algebras;  Integrable systems.}}
\vspace{0cm}

\vspace{3mm}

\section{Introduction}

The Onsager algebra is an infinite dimensional Lie algebra with three known presentations. Introduced by L. Onsager in the investigation of the exact solution of the two-dimensional Ising model \cite{Ons44}, the original presentation is given in terms of generators $\{A_n,G_m|n,m\in {\mathbb Z}\}$ and relations (see Definition \ref{def:OA}). The second presentation is given in terms
of two generators $\{A_0,A_1\}$ satisfying the so-called Dolan-Grady relations (\ref{eq:DG}) \cite{Davies}. Recently \cite{BBC}, a third presentation has been identified. It is given in terms of elements of the non-standard classical Yang-Baxter algebra (\ref{eq:Al}) with r-matrix (\ref{r12basic}).\vspace{1mm}

The Askey-Wilson algebra has been introduced in \cite{Z91}, providing an algebraic scheme for the Askey-Wilson polynomials. This algebra is connected with the double affine Hecke algebra of type $(C_1^{\vee},C_1)$ \cite{K,T12,M,KM}, the theory of Leonard pairs \cite{T87,NT,TV} and $U_q(sl_2)$ \cite{GZ,WZ}.  
A well-known presentation of the Askey-Wilson algebra\footnote{For the {\it universal} Askey-Wislon algebra introduced in \cite{T11}, a second presentation is known.} is given in terms 
of three generators satisfying the relations displayed in Definition \ref{def41}. Generalizations of the Askey-Wilson algebra is an active field of investigation. Various examples of generalizations have been considered in the literature, see for instance \cite{DGVV,GVZ,P,PW}. \vspace{1mm}

In this note, it is shown that the class of quotients of the Onsager algebra considered by Davies in \cite{Davies} generates a classical analog ($q=1$) of the Askey-Wilson algebra and generalizations of this algebra. For each quotient, classical analogs of the automorphisms recently  introduced in \cite{BK17} are used to derive explicit polynomial expressions for the generators.
Based on the results of \cite{BBC} extended to these quotients, for  the classical Askey-Wilson algebra and each of its generalization, a  presentation\ {\it  \`a la Faddeev-Reshetikhin-Takhtajan}  is given. Using this presentation, for each quotient a commutative subalgebra is identified.
To complete the analysis, we also give a new presentation of the Onsager algebra that can be understood as the specialization $q=1$ of the infinite dimensional quantum algebra ${\cal A}_q$ introduced in \cite{BK2}. In this alternative presentation, the quotients of the Onsager algebra corresponding to Davies' prescription are determined. 

\vspace{1mm}

\section{The Onsager algebra, quotients and FRT presentation}
In this section, three different presentations of the Onsager algebra $\cO$ are first reviewed, and three different automorphisms $\Phi,\tau_0,\tau_1$ of the Onsager algebra are introduced. Using these, the elements in $\cO$  are written as
simple polynomial expressions of the fundamental generators $A_0,A_1$. Then, we consider certain quotients of the Onsager algebra introduced by Davies \cite{Davies}. Each quotient is formulated using an operator written as a polynomial in the automorphisms. Given a quotient, the FRT presentation is constructed from which a generating function for mutually commuting quantites is obtained.
\subsection{The Onsager algebra} The Onsager algebra has been introduced in the context of mathematical physics \cite{Ons44}. The first presentation of this algebra which originates in Onsager's work \cite{Ons44} is now recalled. 
\begin{defn}\label{def:OA}
 The Onsager algebra $\cO$ is generated by $\{A_n,G_m|n,m \in\ZZ\}$ subject to the following relations:
 \begin{eqnarray}
 &&[A_n,A_m]=4\ G_{n-m}\;,\label{eq:OA1}\\
 &&[G_n,A_m]=2A_{n+m}-2A_{m-n}\;,\label{eq:OA2}\\
  &&[G_n,G_m]=0\label{eq:OA3}\;.
 \end{eqnarray}
\end{defn} 
\begin{rem} $\{A_n,G_m\}$ for $n\in{\mathbb Z}$ and $m\in {\mathbb Z}_{+}$ form a basis of $\cO$.  Note that $G_{-n}=-G_n$ and $G_0=0$.
\end{rem}
Note that a second presentation is given in terms of two generators $A_0,A_1$ subject to a pair of relations, the so-called
Dolan-Grady relations \cite{DG82}. They read:
 \begin{eqnarray}
[A_0,[A_0,[A_0,A_1]]]=16[A_0,A_1], \qquad [A_1,[A_1,[A_1,A_0]]]=16[A_1,A_0].\label{eq:DG}
 \end{eqnarray}
These two presentations define isomorphic Lie algebras, see \cite{Davies,a-Roan91}.\vspace{1mm}

In a recent paper \cite{BBC}, a third presentation of the Onsager algebra was proposed using the framework of the non-standard classical Yang-Baxter algebra. 
It is called a {\it FRT presentation} in honour of the authors  Faddeev-Reshetikhin-Takhtajan \cite{FRT87}.  Let us introduce the r-matrix  ($u,v$ are formal variables, sometimes called `spectral parameters' in the literature on integrable systems)
\beqa
{r}_{12}(u,v)= \frac{1}{(u-v)(uv-1)}\begin{pmatrix}
       u(1-v^2)&0&0& -2(u-v)\\
       0&-u(1-v^2)& -2v(uv-1)&0\\
       0& -2u(uv-1) & -u(1-v^2) &0\\
        -2uv(u-v)&0&0& u(1-v^2)
      \end{pmatrix}\ 
\label{r12basic}
\eeqa
solution of the non-standard classical Yang-Baxter equation
 \begin{equation}\label{eq:nsCYBE}
  [\ {r}_{13}(u_1,u_3)\ , \ {r}_{23}(u_2,u_3)\ ]=[\ {r}_{21}(u_2,u_1)\ , \ {r}_{13}(u_1,u_3)\ ]+[\ {r}_{23}(u_2,u_3)\ , \ {r}_{12}(u_1,u_2)\ ]\;,
 \end{equation}
where we denote $r_{12}(u) = r(u)\otimes I\!\! I$ , $r_{23}(u) =I\!\! I\otimes  r(u) $  and so on.

\begin{thm} \cite{BBC}\label{th1}
The non-standard classical Yang-Baxter algebra 
\begin{equation}\label{eq:Al}
 [\ B_{1}(u)\ , \ B_{2}(v)\ ]=[\ {r}_{21}(v,u)\ , \ B_{1}(u)\ ]+[\ B_{2}(v)\ , \ {r}_{12}(u,v)\ ]\; 
\end{equation}
for the r-matrix (\ref{r12basic}) and
\begin{equation}
 B(u)=\begin{pmatrix}
       {\cal G}(u) &{\cal A}^-(u)\\
       {\cal A}^+(u) & -{\cal G}(u)
      \end{pmatrix}\label{eq:BO}
\end{equation}
with 
\begin{eqnarray}
{\cal G}(u)=\sum_{n\geq 1} u^n G_{n}\ ,\quad \ {\cal A}^-(u)=\sum_{n\geq 0} u^n A_{-n}
 \ ,\quad \ {\cal A}^+(u)=\sum_{n\geq 1} u^n A_{n}\; ,\label{eq:cu}
\end{eqnarray}
provides an FRT presentation of the Onsager algebra.
\end{thm}

\subsection{Automorphisms of the Onsager algebra}
We are interested in three algebra automorphisms of $\cO$.
Let $\Phi:\cO \rightarrow \cO$ denote the algebra automorphism defined by $\Phi(A_0)=A_1$ and $\Phi(A_1)=A_0$. Observe that $\Phi^2=\id$. We now introduce two other automorphisms of $\cO$.
\begin{prop}\label{prop:aut}
 There exist two involutive algebra automorphisms $\tau_0,\tau_1:\cO\rightarrow \cO$ such that
  \beqa
     \tau_0(A_0)&=&A_0\ ,\label{eq:T00}\\ 
     \tau_0(A_1)&=&-\frac{1}{8}\left(A_1A_0^2 - 2A_0A_1A_0 + A_0^2A_1\right)
         + A_1 =-\frac{1}{8} [A_0,[A_0,A_1]]
         + A_1\ , \label{eq:T01}
  \eeqa
  \beqa
 \tau_1(A_1)&=&A_1\ ,\label{eq:T11}\\ 
     \tau_1(A_0)&=&-\frac{1}{8}\left(A_0A_1^2 - 2A_1A_0A_1 + A_1^2A_0\right) 
         + A_0=-\frac{1}{8} [A_1,[A_1,A_0]]
         + A_0\ .\label{eq:T10}
  \eeqa
\end{prop}
\begin{proof} Firstly, we show that $\tau_0$  leaves invariant the first relation in (\ref{eq:DG}). This follows immediately from the fact that 
\beqa
 [A_0,\tau_0(A_1)]=[A_1,A_0].\label{eq:1}
\eeqa
Secondly, we show that $\tau_0$ leaves invariant the second relation in (\ref{eq:DG}). Observe that:
\beqa
[\tau_0(A_1),[\tau_0(A_1),A_0]] = -8 \tau_0\tau_1(A_0) + 8A_0\ .
\eeqa
It follows:
\beqa
[\tau_0(A_1),[\tau_0(A_1),[\tau_0(A_1),A_0]]] = 8 \underbrace{[\tau_0\tau_1(A_0),\tau_0(A_1)]}_{= \tau_0([\tau_1(A_0),A_1])} + 8 \underbrace{[\tau_0(A_1),A_0]}_{=[A_0,A_1]}= 16[A_0,A_1].\nonumber
\eeqa
So, we conclude that $\tau_0$ leaves invariant both relations in (\ref{eq:DG}). 
\begin{eqnarray}
 &&\tau_0(\tau_0(A_0))=A_0\ ,\\
  &&\tau_0(\tau_0(A_1))=-\frac{1}{8} [A_0,[A_0,\tau_0(A_1)]]+\tau_0(A_1)=A_1\;.
\end{eqnarray}
This proves that $\tau_0$ is involutive and, by consequence is a bijection.

The same holds for $\tau_1$, using $\tau_1=\Phi \circ \tau_0 \circ \Phi$.
\end{proof}
\begin{rem}\label{R1} $(\tau_0\Phi) (\tau_1\Phi) =(\tau_1\Phi) (\tau_0\Phi) =\id$.
\end{rem}

Let us mention that the automorphisms  $\Phi,\tau_0,\tau_1$ can be viewed as the classical analogs $q=1$ of the automorphisms considered in \cite{BK17} (see also \cite{T17}). Using $\tau_0,\tau_1$ and $\Phi$, the elements of the Onsager algebra admit simple expressions as polynomials of the two fundamental generators $A_0,A_1$.
\begin{prop}\label{prop:poly} In the Onsager algebra $\cO$, one has:
\beqa
A_{m}= (\tau_1\Phi)^m(A_0)\  \quad \mbox{and}\quad G_n=\frac{1}{4}[ (\tau_1\Phi)^n(A_0),A_0] \ .\label{eq:idA}
\eeqa
\end{prop}
\begin{proof}
By definition (\ref{eq:OA1}), one has  $G_1=[A_1,A_0]/4$. By Remark \ref{R1}, one has $(\tau_0\Phi)=(\tau_1\Phi)^{-1}$. According to (\ref{eq:T00})-(\ref{eq:T10}), it follows:
\beqa
[G_1,A_0]=2(A_1-\tau_0(A_1))\ , \quad  [G_1,A_1]=2(\tau_1(A_0)-A_0)\ .\label{GA}
\eeqa
Comparing  (\ref{GA}) with (\ref{eq:OA2}), we see that the identification (\ref{eq:idA}) holds for $m=-1,2$. Then, we note that $\tau_1(G_1)=-G_1$ by (\ref{eq:DG}). Acting with $(\tau_1\Phi)^k$ on (\ref{GA}), one derives  (\ref{eq:OA2}) for $n=1$. The second relation in (\ref{eq:idA}) follows from (\ref{eq:OA1}).
\end{proof}

\begin{rem}  $\Phi(A_{-n})=A_{n+1}$ and $\Phi(G_n)=-G_n$.
\end{rem}

In the FRT presentation displayed in Theorem \ref{th1}, the action of the automorphisms  is easily identified. The action of $\tau_0,\tau_1$ on the currents is such that:
\beqa
(\tau_0\Phi)({\cal A}^-(u)) &=& u^{-1}({\cal A}^-(u)-A_0)\ , \quad (\tau_0\Phi)({\cal A}^+(u)) = u({\cal A}^+(u)+A_0)\ ,\label{eq:tau-cu}\\
(\tau_1\Phi)({\cal A}^-(u)) &=& u({\cal A}^-(u)+A_1)\ , \quad \ \ \ (\tau_1\Phi)({\cal A}^+(u)) = u^{-1}({\cal A}^+(u)-A_1)\ ,\nonumber\\
(\tau_0\Phi)({\cal G}(u)) &=& (\tau_1\Phi)({\cal G}(u)) = {\cal G}(u)\ . \nonumber
\eeqa
%
%
%
%


\subsection{Quotients of the Onsager algebra} In Davies' paper on the Onsager algebra and superintegrability \cite{Davies}, Davies considers certain quotients of the Onsager algebra. Below, we characterize the relations considered by Davies in terms of an operator which is a polynomial in two automorphisms $\taub_0,\taub_1$.  As will be shown later, these quotients can be viewed as generalizations of the classical  $(q=1)$ Askey-Wilson algebra.
\begin{defn}\label{def22}  Let $\{\alpha_n|n=0,...,N\}$ be non-zero scalars with $N$ any non-zero positive integer. The algebra $\overline{\cO}_N$
 is defined as the quotient of the Onsager algebra $\cO$ by the relations
\begin{eqnarray}
&&\sum_{n=-N}^{N}\alpha_n A_{-n}=0
\quad\text{and}\qquad\sum_{n=-N}^{N}\alpha_n A_{n+1}=0
\qquad \mbox{with}\quad  \alpha_{-n}=\alpha_{n} \label{dav}\ .
\end{eqnarray}
\end{defn}

There exists an algebra homomorphism $\varphi_N: \cO \rightarrow \overline{\cO}_{N}$ that sends $A_0 \mapsto A_0$,  $A_1\mapsto A_1$.  We now introduce three automorphisms $\taub_0$, $\taub_1$ and $\Phib$ of $\overline{\cO}_{N}$ 
such that $\taub_0\varphi_N = \varphi_N \tau_0$, $\taub_1\varphi_N = \varphi_N \tau_1$ and $\Phib(A_0)=A_1$.
According to Proposition \ref{prop:poly}, introduce the operator:
\beqa
S_N = \sum_{n=-N}^{N}\alpha_n (\taub_1\Phib)^n\ .\label{eq:SN}
\eeqa
The relations in (\ref{dav}) simply read $S_N(A_0)=0$ and $S_N(A_1)=0$, respectively. 
These results allow us to give an alternative presentation of the quotients $\overline{\cO}_{N}$:
\begin{prop}\label{pro:che}
The quotient $\overline{\cO}_{N}$ is generated by $A_0$ and $A_1$ subject to the Dolan Grady relations
 \begin{eqnarray}\label{eq:DGche}
[A_0,[A_0,[A_0,A_1]]]=16[A_0,A_1]\quad \mbox{and}\quad [A_1,[A_1,[A_1,A_0]]]=16[A_1,A_0],
 \end{eqnarray}
and to the relations
\begin{eqnarray}
 S_N(A_0) =0  \quad \mbox{and}\quad S_N(A_1) =0\;,\label{eq:opdav}
\end{eqnarray}
where $S_N$ is defined by \eqref{eq:SN}.
\end{prop}

Furthermore, one has $[(\taub_1\Phib)^p,S_N]=0$ for any $p\in {\mathbb Z}$. Together with the second relation in  (\ref{eq:idA}),  it follows:
\begin{rem} The relations (\ref{dav}) imply:
\beqa
\sum_{n=-N}^{N}\alpha_n A_{n+p}=0\ ,\quad
\sum_{n=-N}^{N}\alpha_n G_{n+p}=0 \quad \mbox{for any} \quad p\in {\mathbb Z}\ .\label{dav2}
\eeqa
It follows that the algebra $\overline{\cO}_N$ has only $3N$ linearly independent elements. We choose the set $\{A_n, G_m|n=-N+1,...,N; m=1,...,N\}$.
\end{rem}
Note that above relations can be derived using the commutation relations (\ref{eq:OA1})-(\ref{eq:OA3})  \cite{Davies}.\vspace{1mm}

In the algebra $\overline{\cO}_N$, all higher elements can be written in terms of the elements $\{A_n, G_m|n=-N+1,...,N; m=1,...,N\}$. Without loss of generality, choose $\alpha_N\equiv 1$. By induction using (\ref{dav2}), one finds:
\beqa
A_{-N-p} = (-1)^{p+N} \sum_{j=-N+1}^N {\mathbb U}_{p,j}^{(N)}(\alpha_0,\cdots,\alpha_{N-1}) A_j \quad \mbox{for any} \quad p\geq 0\ ,
\eeqa
where ${\mathbb U}_{p,j}^{(N)}(\alpha_0,\cdots,\alpha_{N-1})$ is a $N-$variable polynomial that is determined recursively through the relation:
\beqa
{\mathbb U}_{p+1,j}^{(N)}(\alpha_0,\cdots,\alpha_{N-1})  =
 \sum_{k=0}^{p} (-1)^{k}\alpha_{k-N+1}{\mathbb U}_{p-k,j}^{(N)}(\{\alpha_{l}\})\  + \ 
 \begin{cases}
   (-1)^{N+p}\alpha_{j+p+1}&\mbox{for}\ -N+1 \leq j \leq N-p-1 \\
   0 & \mbox{for}\ N-p \leq j \leq N
 \end{cases} \nonumber \ ,
\eeqa 
with the convention $\alpha_{-N+1+k}\equiv 0$ if $k\geq 2N$ and initial conditions:
\beqa
 {\mathbb U}_{0,j}^{(N)}(\alpha_0,\cdots,\alpha_{N-1}) = (-1)^{N+1}\alpha_j \nonumber \ . 
\eeqa
Similarly, one gets:
\beqa
A_{N+p+1} &=& (-1)^{p+N} \sum_{j=-N+1}^N {\mathbb U}_{p,j}^{(N)}(\alpha_0,\cdots,\alpha_{N-1}) A_{1-j}\ ,\nonumber\\
G_{N+p+1} &=& (-1)^{p+N+1} \sum_{j=-N+1}^N {\mathbb U}_{p,j}^{(N)}(\alpha_0,\cdots,\alpha_{N-1}) G_{j-1} \quad \mbox{for any} \quad p\geq 0\ ,\nonumber
\eeqa
where (\ref{eq:OA1}) has been used to derive the second relation. For $N=1$, one finds that ${\mathbb U}_{n-j,j}^{(1)}(\alpha_0)=U_{n}(\alpha_0)$ is the Chebyshev polynomial of second kind.
\vspace{1mm}

\subsection{FRT presentation of the quotients $\overline{\cO}_N$}
For the class of quotients  $\overline{\cO}_N$ of the Onsager algebra, the corresponding solutions of the non-standard Yang-Baxter algebra (\ref{eq:Al}) are now constructed.
\begin{prop}\label{prop1}
 The non-standard classical Yang-Baxter algebra  (\ref{eq:Al}) for the r-matrix (\ref{r12basic}) and 
\begin{equation}
 B^{(N)}(u)= \frac{1}{p^{(N)}(u)}\begin{pmatrix}
       {\cal G}^{(N)}(u) &{\cal A}^{-(N)}(u)\\
       {\cal A}^{+(N)}(u) & -{\cal G}^{(N)}(u)
      \end{pmatrix}  \qquad \mbox{with} \qquad  p^{(N)}(u)=\sum_{p=-N}^{N}\alpha_p u^{-p}\label{eq:BOfinite}
\end{equation}
where, by setting $\displaystyle f_p^{(N)}(u)=\sum_{q=p}^N\alpha_q u^{p-q}$,
\begin{eqnarray}
 {\cal A}^{+(N)}(u)&=& \sum_{p=1}^{N}\big(f_p^{(N)}(u)A_p-uf_p^{(N)}(u^{-1}) A_{-p+1}\big)\ ,\label{fn1}\\
{\cal A}^{-(N)}(u)&=& \sum_{p=1}^{N}\big(u^{-1}f_p^{(N)}(u) A_{-p+1}-f_p^{(N)}(u^{-1})A_p\big)\ ,\label{fn2}\\
{\cal G}^{(N)}(u)&=&    \sum_{p=1}^{N} \big( f_p^{(N)}(u) +f_p^{(N)}(u^{-1}) \big) G_p  - \sum_{p=1}^{N} \alpha_pG_p\   \ ,\label{fn3}
\end{eqnarray}
provides an FRT presentation of the algebra ${\overline \cO_N}$.
\end{prop}
\begin{proof}  The goal consists in expressing all the elements $\{A_n,G_m|n,m \in\ZZ\}$ present in the FRT presentation of the Onsager algebra (see Theorem \ref{th1}) in terms 
of the $3N$ linearly independent elements of $\overline \cO_N$ $\{A_n, G_m|n=-N+1,...,N; m=1,...,N\}$.
For instance, let us consider the current ${\cal A}^{+}(u)$ in (\ref{eq:BO}). Imposing the first relation of (\ref{dav2}), it follows:
\begin{eqnarray}
{\cal A}^{+}(u)&=& \sum_{p=1}^{N}u^pA_p + \sum_{p=N+1}^{\infty}u^{p}A_p\nonumber\\
&=& \sum_{p=1}^{N}u^pA_p -\frac{1}{\alpha_N}\sum_{p=1}^{\infty}u^{p+N}\sum_{q=-N}^{N-1}\alpha_q A_{p+q}\nonumber\\
&=&  \sum_{p=1}^{N}u^pA_p  -\frac{1}{\alpha_N} \sum_{q=-N}^{-1}\alpha_q u^{N-q} \!\!\!\!\!\!  \underbrace{  \sum_{p=1}^{\infty}u^{p+q} A_{p+q}}_{={\cal A}^{+}(u) + \sum_{p=q+1}^0 u^pA_p} \!\!\!\!\!\!  - \frac{\alpha_0}{\alpha_N}u^N 
\underbrace{ \sum_{p=1}^{\infty}u^{p} A_{p}}_{={\cal A}^{+}(u)} 
 -\frac{1}{\alpha_N} \sum_{q=1}^{N-1}\alpha_q u^{N-q}
 \!\!\!\!\!\! \underbrace{ \sum_{p=1}^{\infty}u^{p+q} A_{p+q}}_{={\cal A}^{+}(u) - \sum_{p=1}^q u^pA_p}\ .\nonumber
\end{eqnarray}
By factorizing ${\cal A}^{+}(u)$ in the last equation and after simplifications, one gets:
\beqa
{\cal A}^{+}(u) \underbrace{ \sum_{q=-N}^N \alpha_q u^{N-q}}_{\equiv u^N p^{(N)}(u)}= \sum_{q=1}^{N}\alpha_q u^{N-q} \sum_{p=1}^qu^pA_p  - \sum_{q=-N}^{-1}\alpha_q u^{N-q}\sum_{p=q+1}^0 u^pA_p\ .
\eeqa
It follows:
\beqa
{\cal A}^{+}(u)  = \frac{1}{p^{(N)}(u)} \sum_{q=1}^{N}\sum_{p=1}^q \big( \alpha_q u^{p-q} A_p - \alpha_{-q} u^{q-p+1} A_{-p+1}\big) \ ,\nonumber
\eeqa
which leads to the formula (\ref{fn1}). 
Applying the same procedure to  ${\cal A}^{-}(u)$ and ${\cal G}(u)$, we obtain
the other formulae. 
\end{proof}

Using the FRT presentation, a commutative subalgebra of $\overline{\cO}_N$ can be easily identified. Note that the result below is a straightforward restriction of \cite[Proposition 2.5]{BBC} to the quotients of the Onsager algebra.
\begin{prop}\label{prop220} Let $\kappa,\kappa^*,\mu$ be generic scalars. A generating function of mutually commuting elements in  $\overline{\cO}_N$ is given by:
\beqa
b^{(N)}(u) &=& \frac{1}{p^{(N)}(u)}\ \sum_{p=0}^{N-1} \big(f_p^{(N)}(u)-f_p^{(N)}(u^{-1})\big)I_{p}\ ,\label{gen}
\eeqa
where 
\beqa
I_{p}&=& \kappa (A_p+A_{-p}) + \kappa^* (A_{p+1}+A_{-p+1}) + \mu(G_{p+1}-G_{p-1})\ ,\label{charge} \\ \nonumber
I_0&=&  \kappa A_0+ \kappa^* A_1 + \mu G_{1}\ .
\eeqa
\end{prop}
\begin{proof} Introduce the $2 \times 2$ matrix:
\begin{equation}
  M(x)=\begin{pmatrix}
        \mu/x&\kappa+\kappa^*/x\\
        \kappa + \kappa^* x &\mu x
       \end{pmatrix}\;
 \end{equation}
which is a solution of
\begin{equation}\label{eq:reD}
  [tr_1 ( \overline{r}_{12}(u,v) M_1(u) ) \ ,\ M_2(v) ]=0\; .
 \end{equation}
Then, by using the result \cite[Proposition 2.5 ]{BBC}, one shows that $b^{(N)}(u)=tr M(u) B^{(N)}(u)$
satisfies\ $[ b^{(N)}(u)\ ,\ b^{(N)}(v) ]=0$. 
Inserting  (\ref{fn1})-(\ref{fn3}) in $b^{(N)}(u)=tr \big( M(u) B^{(N)}(u)\big)$, one derives (\ref{gen}).
\end{proof}

\section{$\overline{\cO}_1$ and $\overline{\cO}_2$ and generalized classical Askey-Wilson algebras}

The defining relations of the algebra $\overline{\cO}_N$ are easily extracted from the defining relations of the non-standard classical Yang-Baxter algebra (\ref{eq:Al}).  
For instance, we consider the cases $N=1,2$ below. For $N=1$, we prove that $\overline{\cO}_1$ is isomorphic to the Askey-Wilson algebra introduced by \cite{Z91} specialized at $q=1$.
\subsection{The classical Askey-Wilson algebra $aw(3)$}

We treat here in detail the case of the quotient $\overline{\cO}_1$. To simplify the notations, we choose $\alpha_0=\alpha$ and $\alpha_{\pm 1}=1$. 
Equation \eqref{eq:BOfinite} becomes
\begin{equation}
  B^{(1)}(u)=\frac{1}{p^{(1)}(u)}\begin{pmatrix}
         G_1 & u^{-1}A_0-A_1\\
        -uA_0+A_1 &-G_1
       \end{pmatrix} \  \label{B1}
\end{equation} 
where $p^{(1)}(u)=u+\alpha+u^{-1}$. Then, the non-standard Yang-Baxter algebra (\ref{eq:Al}) provides the following defining relations of $\overline \cO_1$
\begin{eqnarray}
 [G_1\ , \ A_0 ]= 2\alpha A_0 + 4 A_1\ ,\quad \,\quad [A_1\ ,\ G_1]=2\alpha A_1 +4A_0 \ ,\quad \,\quad [A_1\ ,\ A_0]=4G_1\; .\label{eq:O1}
\end{eqnarray}
\begin{rem}  The r-matrix \eqref{r12basic} allows us to construct a representation of $\overline \cO_1$. Indeed, the map $\pi(B^{(1)}_1(u))=r_{13}(u,w)$ satisfies the non-standard Yang-Baxter algebra (\ref{eq:Al}) 
and the expansion w.r.t. $u$ are the same. By comparing the expansions, one gets the following representation, for $\alpha=-w-w^{-1}$,
\begin{equation}\label{eq:rep1}
 \pi(G_1)=(w^{-1}-w)\begin{pmatrix}
           1&0\\0&-1
          \end{pmatrix}
          \quad,\qquad
           \pi(A_0)=2\begin{pmatrix}
           0&1\\1&0
          \end{pmatrix}
          \quad\text{and}\qquad
           \pi(A_1)=2\begin{pmatrix}
           0&w^{-1}\\w&0
          \end{pmatrix}\ .
\end{equation}
\end{rem}
By Proposition \ref{pro:che}, there is another presentation of the algebra $\overline{\cO}_1$. Indeed, $\overline{\cO}_1$ is generated by $A_0$ and $A_1$ subject to
\begin{equation}
 [A_0,[A_0,A_1 ]]- 8\alpha A_0 - 16 A_1=0\ ,\quad \,\quad [A_1,[A_1,A_0]]- 8\alpha A_1 -16A_0=0. \ \label{eq:relaw}
\end{equation}
Let us remark that the Dolan-Grady relations \eqref{eq:DGche} are not necessary in this case since they are implied by \eqref{eq:relaw}.
\vspace{1mm}
 In \cite{Z91}, Zhedanov introduced the Askey-Wilson algebra  with three generators $K_0,K_1,K_2$ and deformation parameter $q$. More recently, a central extension of the original Askey-Wilson algebra \cite{Z91} 
 called the universal Askey-Wilson algebra has been introduced \cite{T11}. In that paper, besides the original presentation of \cite{Z91}, a second presentation of the  universal Askey-Wilson algebra is given. 
 Below, we show that the quotient of the Onsager algebra $\overline{\cO}_1$ is isomorphic to the classical ($q=1$) analog of the Askey-Wilson algebra, denoted $aw(3)$. The first presentation of the original Askey-Wilson algebra is now recalled.
\begin{defn}\label{def41} \cite{Z91}
The Askey-Wilson algebra  has three generators $K_0,K_1,K_2$ that satisfy the commutation relations\footnote{We denote the $q-$commutator $[X,Y]_q=qXY-q^{-1}YX$.}:
\beqa
\big[K_0,K_1\big]_q=K_2\ ,\quad
\big[K_2,K_0\big]_q=BK_0+C_1K_1+D_1\ ,\quad
\big[K_1,K_2\big]_q=BK_1+C_0K_0+D_0\ , \label{eq:aw3K}
\eeqa
where $B,C_0,C_1,D_0,D_1$ are the structure constants of the algebra.
\end{defn}
\begin{rem} In terms of the generators $K_0,K_1$, the defining relations of the Askey-Wilson algebra read:
\beqa
\big[K_0,\big[K_0,K_1\big]_q\big]_{q^{-1}}+BK_0+C_1K_1+D_1=0\ ,\quad
\big[K_1,\big[K_1,K_0\big]_q\big]_{q^{-1}}+BK_1+C_0K_0+D_0=0\ \nonumber.
\eeqa
\end{rem}

\begin{defn}\label{def42}
The classical Askey-Wilson algebra, denoted $aw(3)$, is the Askey-Wilson algebra specialized to $q=1$.
We keep the same notations for the classical Askey-Wilson algebra than for the usual Askey-Wilson algebra.
\end{defn}

\begin{prop}\label{prop41} The algebra $\overline{\cO}_1$ and the algebra $aw(3)$ are isomorphic.
\end{prop}
\begin{proof} The defining relations of $\overline{\cO}_1$ are given in (\ref{eq:O1}). Let $a_0,a_1,b_0,b_1$ be generic scalars. The isomorphism is given by:
\beqa
K_0 \mapsto a_0A_0+b_0 \ ,\quad K_1 \mapsto a_1A_1+b_1 \ ,\quad K_2 \mapsto -\frac{a_0a_1}{4}G_1\   ,\quad q \mapsto 1\nonumber
\eeqa
with the identification of the structure constants:
\beqa
B \mapsto -8\alpha/a_0a_1\ , \quad C_0= -16/a_0^2\ ,\quad C_1=-16/a_1^2\ ,\quad D_0=  -\frac{8\alpha b_0 +16b_1}{a_0^2a_1}\ ,\quad    D_1=-\frac{8\alpha b_1 +16b_0}{a_1^2a_0}\ .\nonumber
\eeqa
\end{proof}
A corollary of this proposition is that Proposition \ref{prop1} provides an FRT presentation of $aw(3)$. Note that for a specialization of the structure constants $B=D_0=D_1$ in (\ref{eq:aw3K}), one recovers the q-deformation of the Cartesian presentation of the $sl_2$ Lie algebra \cite{Z92}. From that point of view, the representation (\ref{eq:rep1}) is natural.\vspace{1mm}

The universal Askey-Wilson algebra has been introduced in \cite{T11}. For this algebra, a second presentation is known \cite[Theorem 2.2]{T11}. It is given in terms of the quotient of 
the $q$-deformed analog of the Dolan-Grady relations (\ref{eq:DG}) by a relation of quartic order in the two fundamental generators.  These relations correspond to the presentation 
given by relations \eqref{eq:relaw}.
Let us mention also that, from the second relation of (\ref{dav2}) with $N=p=1$, one gets \ $\alpha G_1+G_2=0$. In terms of $A_0,A_1$ this relation reads: 
\beqa
8\alpha[A_1,A_0]  + 2(A_1A_0A_1A_0 -A_0A_1A_0A_1)-A_1^2A_0^2 + A_0^2A_1^2 = 0. \label{eq:extra}
\eeqa
Note that (\ref{eq:extra}) is not necessary: 
it follows from the commutator of the first (resp. second) relation in (\ref{eq:relaw}) with $A_0$ (resp. $A_1$).
We would like to point out that the relations  (\ref{eq:relaw}) coincide with (2.2), (2.3) of \cite{T11} for the specialization $q=1$ and central elements evaluated to scalar values. 
Also, the Dolan-Grady relations (\ref{eq:DG}) together with (\ref{eq:extra}) coincide with the specialization  $q=1$  (and a suitable identification of the central element $\gamma$ in terms of $\alpha$)
of the relation given in \cite[Theorem 2.2]{T11}.
\vspace{1mm}

\subsection{The generalized classical Askey-Wilson algebra $aw(6)$}

 For $N=2$, choose $\alpha_0=\alpha'$, $\alpha_{\pm 1}=\alpha$ and $\alpha_{\pm 2}=1$, equation \eqref{eq:BOfinite} reads
\begin{equation}
  B^{(2)}(u)=\frac{1}{p^{(2)}(u)}\begin{pmatrix}
         G_2 +(u+\alpha  +u^{-1})G_1& u^{-1} A_{-1}+u^{-1}(\alpha+u^{-1})A_0-(u+\alpha) A_1- A _2\\
       -u A_{-1}-u(u+\alpha) A_0+u(\alpha + u^{-1})A_1+ A_2 & - G_2 -(u+\alpha +u^{-1})G_1
       \end{pmatrix}
\end{equation}
where $p^{(2)}(u)=u^2+\alpha u +\alpha' +\alpha u^{-1}+u^{-2}$. One gets the following defining relations for $\overline{\cO}_2$ from (\ref{eq:Al})
\begin{eqnarray}
&&[ A_0\ , \ A_{-1}] = [A_2 \ , \ A_1] =[A_1  \ , \ A_0 ] = 4G_1  \ ,\quad [A_1 \ , \ A_{-1}] =[A_2 \ , \ A_0] = 4 G_2  \ ,\quad  \label{aw61}\\
&&  [ A_2 \ , \ A_{-1}] =4(1-\alpha')G_1-4\alpha G_2\ , \\
&&[G_1 \ , \ A_0] = 2A_1 -2 A_{-1} \ ,\quad [G_1 \ , \ A_1] = 2A_2 -2 A_{0}\ ,\\
&& [G_1 \ , \ A_{-1}] =2\alpha A_{-1}+2(1+\alpha') A_0+2\alpha A_1+2 A_2 \ ,\\
&& [G_1 \ , \ A_2] = -2A_{-1} -2\alpha A_0 -2(1+\alpha')A_1-2\alpha A_2 \ ,\\
&& [G_2 \ , \ A_0] =2\alpha A_{-1} +2\alpha' A_0 +2\alpha A_1 +4 A_2\ ,\\
&& [G_2 \ , \ A_1] =-4A_{-1} -2\alpha A_0 -2\alpha' A_1 -2\alpha A_2\ ,\\
&& [ G_2 \ , \ A_{-1} ] =2(\alpha'-\alpha^2) A_{-1}+2\alpha(1-\alpha')A_0+2(2-\alpha^2)A_1-2\alpha A_2\ ,\\
&& [ G_2 \ , \ A_2 ] =2\alpha A_{-1}+2(\alpha^2-2)A_0+2\alpha(\alpha'-1)A_1+2(\alpha^2-\alpha')A_2 \ ,\\
&&[G_1\ , \ G_2 ]=0 \ .\label{aw610}
\end{eqnarray}
\begin{rem}
As previously, a representation of  $\overline{\cO}_2$ is obtained from the r-matrix as follows
\begin{equation}
 \pi(B^{(2)}_1(u))=r_{13}(u,w_1)+r_{14}(u,w_2)
\end{equation}
with $\alpha=-w_1-w_1^{-1}-w_2-w_2^{-1}$ and $\alpha'=w_1w_2+w_1w_2^{-1}+2+w_1^{-1}w_2+w_1^{-1}w_2^{-1}$.
By expanding w.r.t. the formal variable $u$, one gets a $4\times4$ representation for $A_{-1},A_0,A_1,A_2,G_1$ and $G_2$.
\end{rem}

By analogy with the classical Askey-Wilson algebra $aw(3)$ with defining relations \eqref{eq:O1}, we call the algebra generated by the $6$ elements $A_{-1},A_0,A_1,A_2,G_1,G_2$ 
subject to the relations (\ref{aw61})-(\ref{aw610}) the generalized classical Askey-Wilson $aw(6)$. By construction, this algebra is isomorphic to $\overline \cO_2$.

By using Proposition \ref{pro:che} for $N=2$, we get another presentation of the algebra $\overline{\cO}_2\cong aw(6)$ : it is generated by $A_0$ and $A_1$ subject to the Dolan-Grady relation \eqref{eq:DGche}
with the additional following relations
\begin{eqnarray}\label{eq:aw6}
&&[A_0,[A_1, [A_0, [A_1,A_0] ]]] -16 [A_1,[A_1,A_0] ] -8\alpha [A_0,[A_0,A_1] ] +64(\alpha'+2)A_0+128\alpha A_1=0\ ,\label{eq:relaw61}\\
&&[A_1 ,[A_0 , [A_1 , [A_0, A_1] ]]] -16 [A_0, [A_0,A_1] ] -8\alpha [A_1, [A_1,A_0] ] +64(\alpha'+2)A_1+128\alpha A_0=0.\label{eq:relaw62}
\end{eqnarray}
%

By analogy with both previous examples, we define the generalization of the classical Askey-Wilson algebra, denoted $aw(3N)$, as the algebra $\overline{\cO}_N$ generated by $3N$ generators 
$\{A_{-N+1},...,A_{N}\}$ and $\{G_1,...,G_{N}\}$ and subject to the relations projecting the FRT relation \eqref{eq:Al}. The number of defining relations $3N(3N-1)/2$ and 
we do not write them explicitly. Using the the FRT presentation, these relations can be easily extracted.
We can alternatively define $aw(3N)$ with the help of Proposition \ref{pro:che}, as the algebra generated by $A_0,A_1$ and subject to the Dolan-Grady relations (\ref{eq:DGche}) 
and relations \eqref{eq:opdav}. Finally, let us recall that 
a generating function of elements of its commutative subalgebra is given in Proposition \ref{prop220}.

\section{Another presentation of the Onsager algebra and its quotients} 

In this section, a Lie algebra denoted ${\cal A}$ is introduced. 
It is shown to be isomorphic with the Onsager algebra. The corresponding FRT presentation is given, and polynomial expressions for the elements in $\cal A$ are obtained in terms of the two fundamental generators using the automorphisms introduced in Section 2.
Then, we introduce the algebra $\overline{\cal A}_N$ as a quotient of  ${\cal A}$ by the classical analog of the relations derived in \cite{BK2}. The FRT presentation of $\overline{\cal A}_N$ is given.  
\subsection{Another presentation of the Onsager algebra}
In \cite{BK2} (see also \cite{BSh1}), an infinite dimensional quantum algebra denoted ${\cal A}_q$ has been introduced. Recently, it has been conjectured to be isomorphic to the $q-$Onsager algebra\footnote{The q-Onsager algebra is defined in terms of generators and q-analogs of the Dolan-Grady relations (\ref{eq:DG}), see \cite{Ter99}, \cite{Bas05}. Note that the same
relations showed up earlier in the context of polynomial association schemes \cite{Ter93}.} \cite{BB5}. We now introduce the classical analog of  ${\cal A}_q$ ($q=1$).
\begin{defn}\label{def:W}
${\cal A}$ is a Lie algebra with generators $\{{\cW}_{-k},{\cW}_{k+1},\tilde{\cG}_{k+1}|k\in {\mathbb Z}_{\geq 0}\}$ satisfying the following relations, for $k,l\geq 0$:
\begin{align}
&\big[{\cW}_{-l},{\cW}_{k+1}\big]=\tilde{\cG}_{k+l+1}\ ,\label{qo1}\\
&\big[{\tilde{\cG}}_{k+1},{\cW}_{-l}\big]=16{\cW}_{-k-l-1}-16{\cW}_{k+l+1}\ ,\label{qo2}\\
&\big[{\cW}_{l+1},{\tilde{\cG}}_{k+1}\big]=16{\cW}_{l+k+2}-16{\cW}_{-k-l}\ ,\label{qo3}\\
&\big[{\cW}_{-k},{\cW}_{-l}\big]=0\ ,\quad 
\big[{\cW}_{k+1},{\cW}_{l+1}\big]=0\ ,\quad  \big[{\tilde{\cG}}_{k+1},\tilde{{\cG}}_{l+1}\big]=0\ .\label{qo4}\quad 
\end{align}
\end{defn}

\begin{rem} The generators $\cW_0,\cW_1$ satisfy the Dolan-Grady relations \eqref{eq:DG}.
\end{rem}
 Indeed, inserting the relations (\ref{qo1}) into (\ref{qo2}), (\ref{qo3})  for $k=l=0$, from the first two equalities in (\ref{qo4}) for $k=1,l=0$ one gets:
\beqa
[\cW_0,[\cW_0,[\cW_0,\cW_1]]]=16[\cW_0,\cW_1], \qquad [\cW_1,[\cW_1,[\cW_1,\cW_0]]]=16[\cW_1,\cW_0].\label{eq:DGW}
\eeqa

\begin{prop}
\label{prop20}
 The non-standard classical Yang-Baxter algebra  (\ref{eq:Al}) for the r-matrix (\ref{r12basic}) and 
\begin{equation}
 B(u)= \frac{1}{2}\begin{pmatrix}
   -\frac{1}{4}\ \tilde{\cG}(u)    & u^{-1} \cW_+(u) - \cW_-(u)  \\
 -u \cW_+(u) + \cW_-(u)  &      \frac{1}{4}\ \tilde{\cG}(u)    
      \end{pmatrix}\label{eq:BW}
\end{equation}
with, by setting  $U=(u+u^{-1})/2$, 
\beqa
{\cW}_+(u)=\sum_{k=0}^\infty{\cW}_{-k}U^{-k-1} \ , \quad {\cW}_-(u)=\sum_{k=0}^\infty{\cW}_{k+1}U^{-k-1} \ ,\quad  \tilde{\cG}(u)=\sum_{k=0}^\infty\tilde{{\cG}}_{k+1}U^{-k-1}\  , \label{eq:cuW}
\eeqa
provides an FRT presentation of the algebra ${\cal A}$.
\end{prop}
\begin{proof} 
Insert (\ref{eq:BW}) into (\ref{eq:Al}) with (\ref{r12basic}). Define the formal variables $U=(u+u^{-1})/2$ and $V=(v+v^{-1})/2$. One obtains equivalently:
\begin{eqnarray}
&& (U-V)\big[{\cW}_+(u),{\cW}_-(v)\big]=  \tilde{\cG}(v)-\tilde{\cG}(u)\ ,\nonumber\\
&&(U-V)\big[\tilde{\cG}(u),{\cW}_\pm(v)\big]\pm 16\big(U{\cW}_\pm(u)-V{\cW}_\pm(v)-{\cW}_\mp(u)+{\cW}_\mp(v)\big)=0\ ,\nonumber\\
&&\big[{\cW}_\pm(u),{\cW}_\pm(v)\big]=0\ ,\quad \big[\tilde{\cG}(u),\tilde{\cG}(v)\big]=0\ .\nonumber
\end{eqnarray}
Expanding the currents as (\ref{eq:cuW}), the above equations are equivalent to (\ref{qo1})-(\ref{qo4}).
\end{proof}

\begin{thm}\label{th3} The Onsager algebra  $\cO$ (see Definition \ref{def:OA}) and the algebra ${\cal A}$ (see Definition \ref{def:W}) are isomorphic.
\end{thm}
\begin{proof} By Theorem \ref{th1} and Proposition \ref{prop20}, the Onsager algebra $\cO$ and the algebra ${\cal A}$ have the same FRT presentation (\ref{eq:Al}) with the same r-matrix (\ref{r12basic}).
Then, the isomorphism between $\cO$ and ${\cal A}$ follows from the fact that the power series of the entries in (\ref{eq:BO}), (\ref{eq:BW}) have same expansions w.r.t. the formal variable. 
\end{proof}
The explicit relation between the generators $\{A_k,G_l|k\in\mathbb{Z},l\in \mathbb{Z}_{\geq 0}\}$ of the Onsager algebra $\cO$ and the generators $\{\cW_{-k},\cW_{k+1},\tilde{\cG}_{l+1}|k,l\in\mathbb{Z}_{\geq 0}\}$ of 
the algebra ${\cal A}$ is obtained as follows. By comparison between (\ref{eq:BO}) and (\ref{eq:BW}), we get:
\beqa
&&{\cal A}^+(u)\equiv\frac{1}{2}\big( -u \cW_+(u) + \cW_-(u) \big) \ ,\quad {\cal A}^-(u) \equiv \frac{1}{2}\big( u^{-1} \cW_+(u) - \cW_-(u) \big)\ ,\quad
 {\cal G}(u)\equiv -\frac{1}{8}\tilde{\cG}(u) \ \label{eq:isoAW}
\eeqa
with (\ref{eq:cu}) and (\ref{eq:cuW}). Then, one can prove that one has the following expansion around $u=0$, for $k\geq 0$:
\beqa
U^{-k-1}=2\sum_{p=0}^\infty c_{p}^{2p+k} u^{2p+k+1} \quad \mbox{with}\quad c_{p}^k=(-1)^p2^{k-2p}\frac{(k-p)!}{(p)!(k-2p)!}\ . \nonumber
\eeqa
By direct comparison of the l.h.s and r.h.s in (\ref{eq:isoAW}) it follows, for $k\geq 0$,
\beqa
A_{k+1}&=&  \sum_{p=0}^{\left[\frac{k}{2}\right]}  c_p^{k} \cW_{k-2p+1}- \sum_{p=0}^{\left[\frac{k-1}{2}\right]}  c_p^{k-1} \cW_{-k+2p+1}  \ ,\label{eq:A+W}\\
A_{-k}&=&  \sum_{p=0}^{\left[\frac{k}{2}\right]}  c_p^{k} \cW_{2p-k}- \sum_{p=0}^{\left[\frac{k-1}{2}\right]}  c_p^{k-1} \cW_{k-2p}  \ ,\label{eq:A-W}\\
G_{k+1}&=&        - \frac{1}{4} \sum_{p=0}^{\left[\frac{k}{2}\right]}  c_p^{k} \tilde{\cG}_{k-2p+1}\ .\label{eq:GGt}
\eeqa
Conversely, one has:
\beqa
\cW_{-k}&=& \frac{1}{2^k} \sum_{p=0}^{k}  \frac{k!}{p!(k-p)!}A_{k-2p}\ ,\quad\cW_{k+1}= \frac{1}{2^k} \sum_{p=0}^{k}  \frac{k!}{p!(k-p)!}A_{k+1-2p}\ ,\label{eq:W}\\
\tilde{\cG}_{k+1}&=& \frac{1}{2^{k-2}} \sum_{p=0}^{k}  \frac{k!}{p!(k-p)!}G_{2p-k-1} \ .\label{eq:Gt}
\eeqa
Here $[n]$ is the integer part of $n$ (with the convention $[-1/2]=-1$). For small values of $k$, explicit relations between the first few elements are reported in Appendix A. \vspace{1mm}

According to Theorem \ref{th3}, (\ref{eq:W}), (\ref{eq:Gt}) and (\ref{qo4}), the following three lemmas are easily shown.
\begin{lem} The following subsets form a basis for the same subspace of $\cO$:
\beqa
&(i)& A_0,  \quad  A_1+A_{-1}, \quad  A_2+A_{-2},   \quad A_3+A_{-3}, \cdots\nonumber\\
&(ii)& \cW_0, \quad \cW_{-1}, \quad  \cW_{-2},  \quad \cW_{-3}, \cdots \nonumber
\eeqa
\end{lem}

\begin{lem} The following subsets form a basis for the same subspace of $\cO$:
\beqa
&(i)& A_1,\quad    A_2+A_0, \quad   A_3+A_{-1}, \quad  A_4+A_{-2}, \cdots\nonumber\\
&(ii)& \cW_1, \quad \cW_2, \quad \cW_3, \quad\cW_4, \cdots \nonumber
\eeqa
\end{lem}

\begin{lem} The following subsets form a basis for the same subspace of $\cO$:
\beqa
&(i)& G_1, \quad G_2, \quad G_3, \quad G_4,\cdots\nonumber\\
&(ii)&  \tilde{\cG}_1, \quad\tilde{\cG}_2,\quad \tilde{\cG}_3, \quad  \tilde{\cG}_4, \cdots \nonumber
\eeqa
\end{lem}

\subsection{Automorphisms of the algbra $\cal A$}
In view of the isomorphism between $\cO$ and $\cal A$,  the action of the automorphisms $\tau_0,\tau_1,\Phi$  introduced  in Proposition \ref{prop:aut}  is now described  in the alternative presentation $\cal A$. Inverting the correspondence (\ref{eq:isoAW}), one has:
\beqa
{\cal W}_+(u)&\equiv&\frac{2}{(u^{-1}-u)}\big( {\cal A}^+(u) + {\cal A}^-(u) \big) \ ,\quad {\cal W}_-(u)\equiv\frac{2}{(u^{-1}-u)}\big( u^{-1}{\cal A}^+(u) + u{\cal A}^-(u) \big) \ ,\label{eq:isoAWinv}\\
 \tilde{\cG}(u)&\equiv& -8{\cG}(u) \ .\nonumber
\eeqa
Using (\ref{eq:tau-cu}), it yields to:
\beqa
\tau_0({\cal W}_+(u)) &=& {\cal W}_+(u)\ , \quad \tau_0({\cal W}_-(u)) = 2U{\cal W}_+(u) - {\cal W}_-(u) - 2\cW_0\ ,\nonumber\\
\tau_1({\cal W}_-(u)) &=& {\cal W}_-(u)\ , \quad \tau_1({\cal W}_+(u)) = 2U{\cal W}_-(u) - {\cal W}_+(u) - 2\cW_1\ ,\nonumber\\
\tau_0(\tilde{{\cal G}}(u))&=& \tau_1(\tilde{{\cal G}}(u)) = -\tilde{{\cal G}}(u)\ .\nonumber 
\eeqa
Using (\ref{eq:cuW}), it follows:
\begin{prop} The action of the automorphisms $\tau_0,\tau_1$ on the elements of ${\cal A}$ is such that:
\beqa
\tau_0({\cal W}_{-k}) &=& {\cal W}_{-k}\ , \quad \tau_0({\cal W}_{k+1}) = 2{\cal W}_{-k-1}- {\cal W}_{k+1}\ ,\label{t0Wk}\\
\tau_1({\cal W}_{k+1}) &=& {\cal W}_{k+1}\ , \quad \tau_1({\cal W}_{-k}) = 2{\cal W}_{k+2}- {\cal W}_{-k}\ ,\label{t1Wk}\\
\tau_0(\tilde{\cal G}_{k+1}) &=& \tau_1(\tilde{\cal G}_{k+1}) =- \tilde{\cal G}_{k+1}\ .\label{t0Gk}
\eeqa
\end{prop}
From (\ref{qo2}),  (\ref{qo3}) note that 
\beqa
{\cal W}_{-k-1} =\frac{1}{16}[\tilde{\cG}_{k+1},\cW_0] + \cW_{k+1} \ ,\quad  {\cal W}_{k+2} =\frac{1}{16}[\cW_1,\tilde{\cG}_{k+1}] + \cW_{-k} \ .\nonumber
\eeqa
Inserting $\tilde{\cG}_{k+1}= [\cW_0,\cW_{k+1}]$ in the first equation above, from (\ref{t0Wk}), (\ref{t0Gk}) one recovers the classical ($q=1$) analogs of the formulae given in Proposition 7.4 of \cite{T17}. Similarly,  $\tilde{\cG}_{k+1}= [\cW_{-k},\cW_{1}]$ can be inserted into the second equation above  in order to rewrite (\ref{t1Wk}).\vspace{1mm}

Combining above relations, one gets:
\beqa
(\tau_0 + \tau_1)({\cal W}_+(u)) &=& 2U {\cal W}_-(u) - 2\cW_1 \  , \quad (\tau_0 + \tau_1)({\cal W}_-(u)) = 2U {\cal W}_+(u) - 2\cW_0 \ .\nonumber
\eeqa
From the expansions (\ref{eq:cuW}), it follows (note that $\cW_1=\tau_1\Phi(\cW_0)$):
\begin{prop}\label{prop32} In the algebra $\cal A$, one has:
\beqa
\cW_{-k} = \left(\frac{\tau_0\Phi + \tau_1\Phi}{2}\right)^k (\cW_0)\ ,\quad \cW_{k+1} = \left(\frac{\tau_0\Phi + \tau_1\Phi}{2}\right)^k (\cW_1)\ \quad \mbox{and}\quad \tilde{\cG}_{k+1}=\big[\cW_0,  \left(\frac{\tau_0\Phi + \tau_1\Phi}{2}\right)^k (\cW_1)\big] \ . \label{eq:idW}\nonumber
\eeqa
\end{prop}
\begin{rem}   $\Phi(\cW_{-k})=\cW_{k+1}$, $\Phi(\tilde{\cG}_{k+1})=-\tilde{\cG}_{k+1}$.
\end{rem}
Note that the polynomial expressions for the elements $\{\cW_{-k},\cW_{k+1},\tilde{\cG}_{k+1}\}$ computed here using the action of the automorphisms can be viewed as the classical ($q=1$) analogs of the expressions computed in \cite{BB5}, where the elements of the algebra ${\cal A}_q$ are derived as polynomials of the fundamental generators $\cW_0,\cW_1$ satisfying the $q-$deformed version of (\ref{eq:DGW}).
\vspace{1mm}

\subsection{Quotients of the Lie algebra ${\cal A}$ and of the Onsager algebra}
 By analogy with the analysis of the previous section, we now introduce certain quotients of the algebra ${\cal A}$. These quotients can be viewed as the classical analogs of the quotients of algebra ${\cal A}_q$ considered in \cite[eq. 11]{BK2}.
\begin{defn}  Let $\{\beta_n|n=0,...,N\}$ be non-zero scalars with $N$ any non-zero positive integer.
The algebra $\overline{\cal A}_N$ is defined as the quotient of the algebra ${\cal A}$ by the relations
\begin{eqnarray}
&&\sum_{k=0}^{N}\beta_k \cW_{-k}=0 
\quad\text{and}\qquad \sum_{k=0}^{N}\beta_k \cW_{k+1}=0
 \ .\label{davW}
\end{eqnarray}
\end{defn}

According to Proposition \ref{prop32}, introduce the operator:
\beqa
S'_N = \sum_{n=0}^{N}\beta_n (\taub_0\Phib + \taub_1\Phib)^n\ .
\eeqa
Then, eqs. (\ref{davW}) simply read $S'_N(\cW_0)=0$ and $S'_N(\cW_1)=0$, respectively. 
Furthermore, one has $[(\taub_0\Phib + \taub_1\Phib)^p,S'_N]=0$ for any $p\in {\mathbb Z}$. It follows:
\begin{rem} The relations (\ref{davW}) imply:
\beqa
 \sum_{k=0}^{N}\beta_k \cW_{-k-p}=0, 
\quad \sum_{k=0}^{N}\beta_k \cW_{k+1+p}=0,
\quad \sum_{k=0}^{N}\beta_k \tilde{\cG}_{k+1+p}=0\qquad \mbox{for any} \quad p\in {\mathbb Z_{\geq 0}}.\label{davW2b}
\eeqa
The algebra $\overline{\cal A}_N$ has $3N$ generators $\{\cW_{-k}, \cW_{k+1}, \tilde{\cG}_{k+1}|k=0,1,...,N-1\}$.
\end{rem}
Note that above relations (\ref{davW2b}) can be derived using the commutation relations (\ref{qo1})-(\ref{qo3}).\vspace{1mm}

\begin{thm} The algebra $\overline{\cal A}_N$ is isomorphic to the quotient of the Onsager algebra $\overline{\cO}_N$ with the identification
\begin{eqnarray}
&& \beta_{2k}= \frac{2^{2k}}{(2k)!}\sum_{p=k}^{\left[\frac{N}{2}\right]} 2p(-1)^{p-k}\frac{(k+p-1)!}{(p-k)!}   \alpha_{2p} \ , \label{eq:beta-alpha1}\\
 &&\beta_{2k+1}= \frac{2^{2k+1}}{(2k+1)!}\sum_{p=k+1}^{\left[\frac{N+1}{2}\right]} (2p-1)(-1)^{p-k-1}\frac{(k+p-1)!}{(p-k-1)!}   \alpha_{2p-1}\ .  \label{eq:beta-alpha2}
\end{eqnarray}
\end{thm}
\begin{proof} By Theorem \ref{th3},  $\cO$ and $\cal A$ are isomorphic, and the isomorphism is given by (\ref{eq:A+W}), (\ref{eq:A-W}), (\ref{eq:GGt}). 
To show that $\overline{\cal A}_N$ and $\overline{\cO}_N$ are isomorphic, it is necessary and sufficient to show that (\ref{dav}) and (\ref{davW}) are equivalent if
relations (\ref{eq:beta-alpha1})-(\ref{eq:beta-alpha2}) hold. By inserting (\ref{eq:A+W}) and (\ref{eq:A-W}) in (\ref{dav}), one gets equivalently (\ref{davW}) by using 
(\ref{eq:beta-alpha1})-(\ref{eq:beta-alpha2}). 
\end{proof}

The corresponding class of solutions of the non-standard Yang-Baxter algebra (\ref{eq:Al}) is now considered.
\begin{prop}\label{prop3} Let $\{\beta_p|p=0,...,N-1\}$ be non-zero scalars with $N\in{\mathbb N}_{\geq 1}$. Then,
 the non-standard classical Yang-Baxter algebra (\ref{eq:Al}) for the r-matrix (\ref{r12basic}) and 
\begin{equation}
 B^{(N)}(u)= \frac{1}{2\tilde{p}^{(N)}(U)}
\begin{pmatrix}
   -\frac{1}{4}\ \tilde{\cG}^{(N)}(u)    & \ u^{-1} \cW_+^{(N)}(u) - \cW_-^{(N)}(u)  \\
  -u \cW_+^{(N)}(u) + \cW_-^{(N)}(u)  &       \frac{1}{4}\ \tilde{\cG}^{(N)}(u)   
\end{pmatrix} \ ,\quad \tilde{p}^{(N)}(U)=\sum_{p=0}^{N}\beta_p U^{p}\ ,
 \label{eq:BWfinite}
\end{equation}
 where
\begin{eqnarray}
&& \cW_+^{(N)}(u)  =\sum_{k= 0}^{N-1} \tilde f_k^{(N)}(U) \cW_{-k}\ ,\quad  \cW_-^{(N)}(u)  = \sum_{k= 0}^{N-1} \tilde f_k^{(N)}(U) \cW_{k+1}\ ,\quad  
\tilde{\cG}^{(N)}(u)  = \sum_{k= 0}^{N-1} \tilde f_k^{(N)}(U) \tilde{\cG}_{k+1}\; \label{eq:cufinite}
\end{eqnarray}
and 
\begin{eqnarray}
\tilde f_k^{(N)}(U)=\sum_{p=k+1}^{N}\beta_p U^{p-k-1}, \label{fqn}
\end{eqnarray}
%
provides an FRT presentation of the algebra $\overline{\cal A}_N$.
\end{prop}

\begin{proof} The proof is similar to the one of Proposition \ref{prop1} by replacing the relations \eqref{eq:BO} and \eqref{dav} by  (\ref{eq:BW}) and (\ref{davW}).
\end{proof}

\begin{rem}  Note that (\ref{eq:BWfinite}) can be interpreted as the classical analog of the Sklyanin's operators constructed in \cite{BK2} satisfying the reflection algebra.
\end{rem}

\vspace{0.5cm}

\noindent{\bf Acknowledgements:} We thank S. Belliard for discussions, P. Terwilliger and A. Zhedanov for comments and suggestions.  P.B.  and N.C. are supported by C.N.R.S. N.C. thanks the IDP for hospitality, where part of this work has been done.
\vspace{0.2cm}

\begin{appendix}
\section{}
From  (\ref{eq:A+W})-(\ref{eq:GGt}), for $k=0,1,2$  one has:
\beqa
A_0&=& \cW_0\ ,\quad A_1=\cW_1\ ,\quad G_1=-\frac{1}{4}\tilde{\cG}_{1}\ ,\nonumber\\
A_{-1}&=& 2\cW_{-1}-\cW_1\ ,\quad A_{2}=2\cW_{2}-\cW_0\ ,\quad G_2=-\frac{1}{2}\tilde{\cG}_{2}\ ,\nonumber\\
A_{-2}&=& 4\cW_{-2}-\cW_0-2\cW_{2}\ ,\quad A_{3}=4\cW_{3}-\cW_1-2\cW_{-1}\ ,\quad G_3=-\tilde{\cG}_{3} + \frac{1}{4}\tilde{\cG}_{1}\ .\nonumber
\eeqa
Conversely, from (\ref{eq:W})-(\ref{eq:Gt}) for $k=1,2$  one has:
\beqa
 \cW_{-1}&=& \frac{A_1 + A_{-1}}{2}\ , \quad \cW_{2}= \frac{A_0 + A_{2}}{2}\ ,\quad \tilde{\cG}_{2}=-2G_2\ ,\nonumber \\
\cW_{-2}&=& \frac{A_2 + 2A_{0} + A_{-2}}{4}\ , \quad \cW_{2}= \frac{A_3 + 2A_{1} + A_{-1}}{4},\quad \tilde{\cG}_{3}=-G_3 -2G_1\ .\nonumber
\eeqa

\end{appendix}

 \end{document}